\documentclass[10pt]{article}   	
\usepackage{geometry,hyperref}                		
\usepackage{graphicx}					
\usepackage{amssymb,mathrsfs,stmaryrd,amsthm,bussproofs,xcolor,amsthm,amsmath}

\newcommand{\cyl}{\textsf{C}}

\newcommand{\midd}{\; \; \mbox{\Large{$\mid$}}\;\;}

\newcommand{\Era}{\rightarrowtail}
\newcommand{\Ela}{\leftarrowtail}

\newcommand{\NP}{\mathbf{NP}}

\newcommand{\shSat}{\sharp\Sat}

\newcommand{\CPL}{\textbf{\textsf{CPL}}}
\newcommand{\CPLc}{\textbf{\textsf{CPL}}_0}

\newcommand{\PL}{\textbf{\textsf{PL}}}
\newcommand{\QPL}{\textbf{\textsf{QPL}}}
\newcommand{\MQPA}{\textbf{\textsf{MQPA}}}

\newcommand{\BOX}{\mathbf{C}}
\newcommand{\DIA}{\mathbf{D}}
\newcommand{\BBOX}{\mathbb{C}}
\newcommand{\BDIA}{\mathbb{D}}

\newcommand{\twoOm}{\Bool^\Nat}
\newcommand{\Nat}{\mathbb{N}}

\newcommand{\fone}{F}
\newcommand{\ftwo}{G}
\newcommand{\fthree}{H}

\newcommand{\model}[1]{\llbracket #1\rrbracket}
\newcommand{\longv}[1]{}

\newcommand{\bone}{\textsl{b}}
\newcommand{\btwo}{\textsl{c}}
\newcommand{\bthree}{\textsl{d}}
\newcommand{\bfour}{\textsl{e}}
\newcommand{\bvar}{\textsl{x}}

\newcommand{\Sat}{\mathtt{SAT}}
\newcommand{\Taut}{\mathtt{TAUT}}
\newcommand{\CSat}{\mathtt{CSAT}}
\newcommand{\CTaut}{\mathtt{CTAUT}}

\newtheorem{prop}{Proposition}
\newtheorem{lemma}{Lemma}
\newtheorem{ex}{Example}
\newtheorem{cor}{Corollary}

\theoremstyle{definition}
\newtheorem{defn}{Definition}
\newtheorem{notation}{Notation}

\newcommand{\start}{\bullet}

\newcommand{\Bool}{\mathbb{B}}

\title{\textbf{Some Remarks on Counting Propositional Logic}\thanks{I wish to thank U. Dal Lago and P. Pistone for their constant guidance and help.
(Of course, I am the only responsible for any mistakes in it.)}}

\author{M. Antonelli}
\date{(part of a joint work with U. Dal Lago and  P. Pistone)}

\begin{document}
\maketitle

\begin{abstract}
Counting propositional logic was recently introduced
in relation to randomized computation
and shown able to logically characterize the full
counting hierarchy~\cite{ICTCS}.
In this paper we aim to clarify the intuitive meaning
and expressive power of its univariate fragment.
On the one hand, we make the connection between
this logic and stochastic experiments
explicit, proving that the counting language 
can simulate \emph{any}
(and only)
event associated with dyadic distributions.
On the other,  we provide an effective
procedure to measure the probability of counting
formulas.
\end{abstract}


\section{Introduction}
The need for reasoning about uncertain knowledge and 
probability has come out
in several areas of research,
from AI to economics, from linguistics to theoretical computer science
(TCS, for short).
For example, probabilistic models are crucial when considering
randomized programs and algorithms or 
dealing with partial information, e.g.~in expert systems.
It was this concrete demand that led to the first attempts 
to analyze probabilistic reasoning \emph{formally}, and
to the development of a few logical systems,
starting in 1986 with Nilsson's pioneering (modal) proposal:
%

\small
\begin{quote}
Because many artificial intelligence applications require
 the ability to reason with uncertain knowledge,
 it is important to seek appropriate generalizations of logic
 from this case. \cite[p. 71]{Nilsson86}
\end{quote}
\normalsize
For probabilistic algorithms behavioral properties,
like termination or equivalence, 
have \emph{quantitative} nature, that is
computation terminates \emph{with a certain probability}, or
programs simulate the desired function \emph{up to some probability} of error, for instance with learning algorithm.
How can such properties be studied within a logical system?
In a series of recent works~\cite{CiE,ICTCS,ADLP22}, we
introduce logics with counting and measure quantifiers,
providing a new formal framework to study probability,
and show them strongly related to several  aspects of randomized computation.
Specifically, counting formulas can be seen as expressing
that a program behaves in a certain way with a given probability.
%
%
%
In this short paper, we aim to clarify what is 
the expressive power 
(and limit) of the simple, univariate fragment of counting propositional logic~\cite{ICTCS},
to better understand its connection 
with both randomized computation
and other probability systems.

\subsection{On Logic and Randomized Computation}

The development of counting logics is part of an overall study aiming at analyzing the interaction between (quantitative) logic and probabilistic computation, so to deepen our knowledge of both.
This project was motivated by two main considerations. 
On the one hand, since their 
appearance
in the 1970s, probabilistic computational models 
have become more and more pervasive in several fast-growing areas
of computer science and technology, such as statistical learning or approximate
computing.
On the other, the development of (deterministic) 
models has considerably benefitted from
interchanges between logic and TCS.
Nevertheless, there is at least one crucial
aspect of the theory of computation
which was only marginally touched by such fruitful
interactions, namely \emph{randomized}
computation.
The global purpose of our study is to lay the
foundation for a new approach to bridge this gap,
and its
key ingredient consists
in considering new \emph{inherently quantitative} logics, 
extended with non-standard quantifiers able to
``measure'' the probability of their argument formula.

So far, we have mostly focussed on a few specific aspects
of the interaction between quantitative logics and randomized computation: 
\begin{itemize}
\itemsep0em

\item[$*$] Complexity theory: 
it is well-known that 
classical propositional logic ($\PL$, for short)
provides the first example of an $\NP$-complete problem~\cite{Cook71},
while its quantified version 
characterizes the full polynomial hierarchy~\cite{MeyerStockmeyer72,MeyerStockmeyer73,BB09}.
Yet, no analogous logical counterpart was found 
for the probabilistic and counting classes~\cite{Gill77,Valiant79} and hierarchy~\cite{Wagner}.
In~\cite{ICTCS}, we introduce a counting logic, called $\CPL$, which is shown to be the probabilistic counterpart of 
quantified propositional logic ($\QPL$, for short).
Indeed, its formulas in a special prenex normal form, characterizes the corresponding level of Wagner's hierarchy.

%

\item[$*$] Programming language theory: 
type systems for randomized $\lambda$-calculi, 
also guaranteeing various forms of termination properties, 
were introduced in the last decades, e.g. in~\cite{SahebDjaromi,JonesPlotkin,DeLiguoroPiperno,DLGH},
but these systems are not ``logically oriented'' and no
Curry-Howard correspondence~\cite{Curry42,Howard80,SorensenUrzyczyn} is known for them.
In~\cite{ADLP22}, we define an intuitionistic counting logic,
which captures quantitative behavioral properties
and typed $\lambda$-calculi in which types 
reveal the actual probability of termination.
In this way, we also provides a probabilistic version 
of the correspondence.

\item[$*$] Computation theory: 
arithmetics and deterministic computation are linked
by deep theorems from logic and recursion theory.
In~\cite{CiE}, we present a quantitative extension for 
the language of Peano Arithmetic,
able to formalize basic results from 
probability theory which are not expressible in standard arithmetic.
We also generalize classical theorems from recursion theorem to the quantitative realm,
for instance due to our randomized version of G\"odel's arithmetization~\cite{Godel31}.
\end{itemize}


%

\subsection{The Structure of the Paper}
We try to clarify the intuitive meaning associated with
our counting logic and the nature of its non-standard 
quantifiers, focussing on a few specific topics.
In particular, the presentation is structured as follows.
First, in Section~\ref{sec:CPL}
we briefly recap some crucial aspects of
measure theory and the semantics of $\CPLc$.
Then, in Section~\ref{sec1} we consider the relation
between this logic and stochastic experiments explicit.
Specifically, we prove that $\CPLc$ can simulate
\emph{any} event associated with dyadic distributions.
Finally, in Section~\ref{sec:measuring} we provide an effective
procedure to measure the probability of counting formulas.
%

\newcommand{\head}{\textsc{head}}
\newcommand{\tail}{\textsc{tail}}

\newcommand{\zero}{\mathbf{0}}
\newcommand{\one}{\mathbf{1}}

\section{On (Univariate) Counting Propositional Logic}\label{sec:CPL}

In order to avoid clash in terminology,
we start by briefly recapping a few notions from basic
probability theory.
Then, we summarise the crucial aspects 
of our non-standard, univariate counting  propositional logic
as first introduced in~\cite{ICTCS}.

\subsection{Preliminaries}

\paragraph{Probability Space.}
In probability theory, an \emph{outcome} or \emph{point} 
is the result of a single
execution of an experiment,
the \emph{sample space} $\Omega$ 
is the set of all possible outcomes, and
an event is then a subset of $\Omega$.
%
%
Two events, say $E_1$ and $E_2$,
are \emph{disjoint} or \emph{mutually exclusive},
when they cannot happen at the same time,
that is $E_1\cap E_2=\emptyset$.
Two events are \emph{(stochastically) independent}
when the occurrence of one does not affect the probability
for the other to occur.
%
A class $\mathscr{F}$ of subsets of $\Omega$ is a ($\sigma$-)\emph{field}
if it contains $\Omega$ itself and
is closed under the formation of complements
and (in-)finite unions.
The largest $\sigma$-field in $\Omega$
is the power class $2^\Omega$ consisting
of \emph{all} the subsets of $\Omega$.
Given a $\sigma$-field $\mathscr{F}$, we call
the \emph{$\sigma$-field generated by $\mathscr{F}$},
denoted $\sigma(\mathscr{F})$, the smallest 
$\sigma$-algebra containing $\mathscr{F}$.
A probability measure $\textsc{Prob}(\cdot)$, 
is a real-valued function
defined on a field $\mathscr{F}$ and
satisfying Kolmogorov's axioms, i.e.~associating 
each event $E$ in the field with
a number $\textsc{Prob}(E)$ so that:
(i)  for each $E\in \mathscr{F}$,
$0\leq \textsc{Prob}(E) \leq 1$,
(ii) $\textsc{Prob}(\emptyset)=0$ and
$\textsc{Prob}(\twoOm)=1$,
(iii) if $E_1,E_2,\dots \in\mathscr{F}$ is a sequence of
disjoint events,
then $\textsc{Prob}\big(\bigcup^\infty_{k=1} E_k\big) =
\sum^{\infty}_{k=1} \textsc{Prob}(E_k)$.
Given two disjoint events, $E_1$ and $E_2$,
$\textsc{Prob}(E_1\cup E_2)= \textsc{Prob}(E_1)
+ \textsc{Prob}(E_2)$,
while for two independent events, $E_1'$
and $E_2'$, $\textsc{Prob}(E_1'\cap E_2') =
\textsc{Prob}(E_1') \cdot \textsc{Prob}(E_2')$.

\paragraph{Cylinder Measure.}
In the following,
we will deal with a specific probability space,
as defined in~\cite{Billingsley}, where
Bilingsley considers a model to simultaneously
fit random drawing of points from a segment
and infinte sequence of coin tosses
(so to be interesting for both geometry and probability).
In particular, when tossing a coin, the
set of the possible outcomes of the 
experiment is $S=\{\tail,\head\}$.
More in general, when dealing with a Bernoulli
experiment, we can consider the set of its
possible outcomes as simply $\Bool=\{\zero,\one\}$
(or even $2=\{0,1\}$).\footnote{In what follows, 
we use all these three notations
basing on the pertinence with the context.
In particular, we use $\tail$ and $\head$ when dealing with
concrete examples 
concerning tossing. 
Of course, all these sets
are equivalent for our goal.}
The corresponding sample space is $\Omega=
\Bool^\Nat$, i.e.~the set of all infinite sequences of random bits
(that is coin tosses) denoted
as $\omega=\omega(1)\omega(2)\dots$,
where for any $\omega\in \Omega$
and $i\ge \Nat$, $\omega(i)\in\Bool$.
\footnote{Notice that we've slightly modified
Billinglsey's notation, where e.g.~the 
infinite-dimensional Cartesian product $\Omega$ as $S^\infty$
and $\omega$ as $z_1(\omega)$.}
Each sequence $\omega$ can be interpreted as the
result of infinitely flipping a coin.

\begin{defn}[Cylinder of Rank $n$]\label{df:cylRank}
A \emph{cylinder of rank n} is a set of
the form $\cyl_H=\{\omega \ | \ \omega(1), \dots, \omega(n) \in H\}$,
with $H\subset \Bool^n$. 
\end{defn}
\noindent
When $H$ is a singleton,
an event $E=\{ \omega \ | \ \omega(1), \dots, \omega(n)
=(u_1,\dots, u_n)\}$,
such that the first $n$ repetitions of the experiment
give the outcomes $u_1,\dots, u_n$ in sequence
is called a \emph{thin cylinder}.

The class of cylinders of all ranks is denoted by
$\mathscr{C}_0$, 
while $\mathscr{C}$ is a field closed under
complementation and union.
It is thus possible to define
a measure on it. 
In particular, a canonical one, consists
in assigning the following probability measure $\mu_{\mathscr{C}}$, to any cylinder of rank $n$.

\begin{defn}[Cylinder Measure]\label{df:measureCyl}
Given $u\in \Bool$, $p_u$ denote the (non-negative
and summing to 1) probabilities
of getting $u$.
For any cylinder $\cyl_H$
$$
\mu_\mathscr{C}(\cyl_H)=\sum_H p_{u_1}
\cdots p_{u_n}.
$$
\end{defn}
\noindent
In the special case of $\cyl_H$
being a thin cylinder 
$
\mu_{\mathscr{C}}\big(\{\omega \ | \ (\omega(1), \dots, \omega(n))
= (u_1,\dots, u_n)\}\big) = p_{u_1} \cdots p_{u_n},
$
providing a model for an infinite sequence of random
bits or independent tosses,
each with probability $p_{\zero_i}$ of success
and $p_{\one_i}$ of failure.
Observe also that when the coin is fair, for each tossing 
$p_{\zero}=p_{\one} = \frac{1}{2}$.
In this case, since cylinders of rank $n$
are finite sets,
the following result is a straightforward
consequence of Definition~\ref{df:cylRank}.

\begin{cor}\label{cor:CylDyad}
For any cylinder of rank $n$,
call it $\cyl_H$, and $p_\zero=p_\one=\frac{1}{2}$,
there are some $l,m\in \Nat$
such that 
$
\mu_\mathscr{C}(\cyl_H) = \frac{l}{2^m}.
$
\end{cor}
\noindent
Finally, going back to $\sigma(\mathscr{C})$,
a well-defined probability measure can be assigned
to it by simply generalizing Definition~\ref{df:measureCyl} in the natural way.
Then, the probability space 
$(\Bool^\Nat, \sigma(\mathscr{C}), \mu_{\mathscr{C}})$,
where $\mu_{\mathscr{C}}$ is 
such that, $p_{\zero}=p_{\one}=\frac{1}{2}$,
defines a standard model for infinite
and independent tosses of a \emph{fair} coin.

\subsection{(Univariate) Counting Propositional Logic in a Nutshell}

\paragraph{Grammar and Semantics.}
In standard $\PL$ formulas are
interpreted as single truth-values.
The core idea of our counting semantics
consists in modifying this intuition in a quantitative sense,
associating formulas with \emph{measurable} sets of
(satisfying) valuations.
Given a counting formula $\fone$,
its interpretation is the set $\model{\fone}\subseteq \twoOm$,
made of all maps $f\in\twoOm$
 ``making $\fone$ true''.
Such sets belong to the standard Borel algebra over $\twoOm$,
$\mathscr{B}(\twoOm)$,
yielding a genuinely quantitative semantics.
Specifically, atomic propositions correspond to 
\emph{cylinder sets}~\cite{Billingsley} of the form:
$$
Cyl(i) = \{f \in \twoOm \ | \ f(i) = \one\},
$$
with $i\in\Nat$.
 Molecular expressions
 are interpreted in a natural way, via
 standard operations of complementation, 
 finite intersection and union over the $\sigma$-algebra.
So, formulas are all measurable and, 
in particular, associated with the
unique cylinder measure  $\mu_{\mathscr{C}}$,
where
$\mu_{\mathscr{C}}(Cyl(i))= \frac{1}{2}$
for any $i\in\Nat$~\cite{Billingsley}.

We enrich this language
with new formulas expressing the measure of such sets.
By adapting the notion of counting operator by 
Wagner~\cite{Wagner},
we introduce two non-standard quantifiers $\BOX^q$ and 
$\DIA^q$, with $q\in \mathbb{Q}_{[0,1]}$.
Then, quantified formulas $\BOX^q \fone$ and $\DIA^q\fone$
express that $\fone$ is satisfied in a certain portion of all its possible
interpretations
to be (resp.) greater or strictly smaller than the given one.
For example, the formula $\BOX^{1/2}\fone$
says that $\fone$ is satisfied by \emph{at least}
half of its valuations.
Semantically, this amounts at checking that
$\mu_{\mathscr{C}}(\model{\fone}) \ge \frac{1}{2}$.

\begin{defn}\label{def:semantics}
\emph{Formulas of $\CPLc$} are defined by the grammar below:
$$
\fone := \mathbf{i} \midd \neg \fone \midd \fone \wedge \fone
\midd \fone \vee \fone \midd \BOX^q \fone \midd \DIA^q\fone,
$$
where $i\in \Nat$ and $q\in \mathbb{Q}_{[0,1]}$.
Given the standard cylinder space $\mathscr{P}
=(\twoOm, \sigma(\mathscr{C}), \mu_{\mathscr{C}})$,
for each formula of $\CPLc$, $\fone$, its \emph{interpretation}
is the measurable set $\model{\fone} \in \mathscr{B}(\twoOm)$
defined as follows:

\begin{minipage}{\linewidth}
\begin{minipage}[t]{0.4\linewidth}
\begin{align*}
\model{\mathbf{i}} &= Cyl(i) \\
\model{\neg \ftwo} &= \twoOm - \model{\ftwo} \\
\model{\ftwo_1 \wedge \ftwo_2} &=
\model{\ftwo_1} \cap \model{\ftwo_2} \\
\model{\ftwo_1\vee \ftwo_2} &=
\model{\ftwo_1} \cup \model{\ftwo_2} 
\end{align*}
\end{minipage}
\hfill
\begin{minipage}[t]{0.55\linewidth}
\begin{align*}
\model{\BOX^q\ftwo} &= \begin{cases}
\twoOm \ &\text{if } \mu_{\mathscr{C}}(\model{\ftwo})
\ge q \\
\emptyset \ &\text{otherwise}
\end{cases} \\
\model{\DIA^q \ftwo} &=
\begin{cases}
\twoOm \ &\text{if } \mu_{\mathscr{C}}(\model{\ftwo}) < q\\
\emptyset \ &\text{otherwise.}
\end{cases}
\end{align*}
\end{minipage}
\end{minipage}
\end{defn}
\noindent
A formula of $\CPLc$ $\fone$, is said \emph{valid}
when $\model{\fone}=\twoOm$
while is said \emph{invalid} when $\model{\fone}=\emptyset$.

\paragraph{Proof Theory.}
In~\cite{ICTCS}, 
we even define a sequent calculus, 
called $\mathbf{LK}_{\CPLc}$, 
which is proved sound and complete for the semantics above.
Its language is labelled, that is its formulas contains 
both a counting and a Boolean part.

\begin{defn}[Boolean Formula]
\emph{Boolean formulas} are defined by the
grammar below:
$$
\bone := \bvar_i \midd \top \midd \bot
\midd \neg \bone \midd \bone \wedge \bone
\midd \bone \vee \bone,
$$
where $i\in\Nat$.
The \emph{interpretation of a Boolean
formula $\bone$}, $\model{\bone}
\in \mathscr{B}(\twoOm)$
is inductively defined as follows:

\begin{minipage}{\linewidth}
\begin{minipage}[t]{0.4\linewidth}
\begin{align*}
\model{\bvar_i} &= Cyl(i) \\
\model{\top} &= \twoOm \\
\model{\bot} &= \emptyset
\end{align*}
\end{minipage}
\hfill
\begin{minipage}[t]{0.5\linewidth}
\begin{align*}
\model{\neg\bone} &= 
\twoOm - \model{\bone}  \\
\model{\bone \wedge \btwo} &=
\model{\bone} \cap \model{\btwo} \\
\model{\bone \vee \btwo} &=
\model{\bone} \cup \model{\btwo}.
\end{align*}
\end{minipage}
\end{minipage}
\end{defn}
\noindent
%
%
So, sequents of $\mathbf{LK}_{\CPLc}$
are made of \emph{labelled expressions}
of the form $\bone \Era \fone$ and
$\bone \Ela \fone$,
where $\bone$ and $\fone$ are (resp.)
a Boolean and a counting formula.
Intuitively, a labelled formula
$\bone \Era \fone$ (resp. $\bone \Ela \fone$)
is true when the set of valuations
satisfying $\bone$ is included
in
(resp. includes) the interpretation of $\fone$.
We use $\bone \vDash \btwo$ for
$\model{\bone} \subseteq \model{\btwo}$.
Then, a \emph{labelled sequent} is
a sequent of the form $\vdash L$,
where $L$ is a labelled formula.
The proof system for $\CPLc$
is defined by the rules illustrated
in Figure~\ref{fig:CPLcproof}.
Observe that some of the rules include
so-called \emph{exteral hypotheses},
i.e.~formulas such as $\bone\vDash \btwo$
or $\mu(\model{\bone})
\triangleright q$, where $\triangleright \in
\{\ge,>, \leq,<,=\}$, $\bone$ and $\btwo$ are
Boolean formulas and 
$q\in\mathbb{Q}_{[0.1]}$.
These hypotheses express semantic
properties of Boolean formulas
or conditions to be checked inside 
$\mathscr{B}(\twoOm)$.

\begin{figure}[h!]
\begin{center}
\framebox{
\parbox[t][12.4cm]{11cm}{
\begin{center}
Initial Sequents

\small
\begin{minipage}{\linewidth}
\begin{minipage}[t]{0.5\linewidth}
\begin{prooftree}
\AxiomC{$\bone \vDash \bvar_n$}
\RightLabel{$Ax1$}
\UnaryInfC{$\vdash \bone \Era \mathbf{n}$}
\end{prooftree}
\end{minipage}
\hfill
\begin{minipage}[t]{0.5\linewidth}
\begin{prooftree}
\AxiomC{$\bvar_n \vDash \bone$}
\RightLabel{$Ax2$}
\UnaryInfC{$\vdash \bone \Ela \mathbf{n}$}
\end{prooftree}
\end{minipage}
\end{minipage}
\end{center}

\normalsize
\begin{center}
Set Rules
\end{center}

\small
\begin{prooftree}
\AxiomC{$\vdash \bthree \Era \fone$}
\AxiomC{$\vdash \bfour \Era \fone$}
\AxiomC{$\bone \vDash \btwo \vee \bthree$}
\RightLabel{$R^\Era_\cup$}
\TrinaryInfC{$\vdash \bone \Era \fone$}
\end{prooftree}

\begin{prooftree}
\AxiomC{$\vdash \btwo \Ela \fone$}
\AxiomC{$\vdash \bthree \Ela \fone$}
\AxiomC{$\btwo \wedge \bthree\vDash \bone$}
\RightLabel{$R^\Ela_\cap$}
\TrinaryInfC{$\vdash \bone\Ela \fone$}
\end{prooftree}

\normalsize
\begin{center}
Logical Rules
\end{center}

\small
\begin{minipage}{\linewidth}
\begin{minipage}[t]{0.45\linewidth}
\begin{prooftree}
\AxiomC{$\vdash \btwo \Ela \fone$}
\AxiomC{$\bone \vDash \neg \btwo$}
\RightLabel{$R^\Era_\neg$}
\BinaryInfC{$\vdash \bone \Era \neg \fone$}
\end{prooftree}
\end{minipage}
\hfill
\begin{minipage}[t]{0.5\linewidth}
\begin{prooftree}
\AxiomC{$\vdash \btwo \Era \fone$}
\AxiomC{$\neg \btwo \vDash \bone$}
\RightLabel{$R^\Ela_\neg$}
\BinaryInfC{$\vdash \bone \Ela \neg \fone$}
\end{prooftree}
\end{minipage}
\end{minipage}

\begin{minipage}{\linewidth}
\begin{minipage}[t]{0.45\linewidth}
\begin{prooftree}
\AxiomC{$\vdash \bone \Era \fone$}
\RightLabel{$R1^\Era_\vee$}
\UnaryInfC{$\vdash \bone \Era \fone \vee \ftwo$}
\end{prooftree}
\end{minipage}
\hfill
\begin{minipage}[t]{0.5\linewidth}
\begin{prooftree}
\AxiomC{$\vdash \bone \Era \ftwo$}
\RightLabel{$R2^\Era_\vee$}
\UnaryInfC{$\vdash \bone \Era \fone \vee \ftwo$}
\end{prooftree}
\end{minipage}
\end{minipage}

\begin{minipage}{\linewidth}
\begin{minipage}[t]{0.45\linewidth}
\begin{prooftree}
\AxiomC{$\vdash \bone \Ela \fone$}
\AxiomC{$\vdash \bone \Ela \ftwo$}
\RightLabel{$R^\Ela_\vee$}
\BinaryInfC{$\vdash \bone \Ela \fone \vee \ftwo$}
\end{prooftree}
\end{minipage}
\hfill
\begin{minipage}[t]{0.5\linewidth}
\begin{prooftree}
\AxiomC{$\vdash \bone\Era \fone$}
\AxiomC{$\vdash \bone \Era \ftwo$}
\RightLabel{$R^\Era_\wedge$}
\BinaryInfC{$\vdash \bone \Era \fone \wedge \ftwo$}
\end{prooftree}
\end{minipage}
\end{minipage}

\begin{minipage}{\linewidth}
\begin{minipage}[t]{0.45\linewidth}
\begin{prooftree}
\AxiomC{$\vdash \bone \Ela \fone$}
\RightLabel{$R1^\Ela_\wedge$}
\UnaryInfC{$\vdash \bone \Ela \fone \wedge \ftwo$}
\end{prooftree}
\end{minipage}
\hfill
\begin{minipage}[t]{0.5\linewidth}
\begin{prooftree}
\AxiomC{$\vdash \bone \Ela \ftwo$}
\RightLabel{$R2^\Ela_\wedge$}
\UnaryInfC{$\vdash \bone \Ela \fone \wedge \ftwo$}
\end{prooftree}
\end{minipage}
\end{minipage}

\normalsize
\begin{center}
Counting Rules 
\end{center}
\small

\begin{minipage}{\linewidth}
\begin{minipage}[t]{0.45\linewidth}
\begin{prooftree}
\AxiomC{$\mu(\model{\bone})=0$}
\RightLabel{$R^\Era_\mu$}
\UnaryInfC{$\vdash \bone \Era \fone$}
\end{prooftree}
\end{minipage}
\hfill
\begin{minipage}[t]{0.5\linewidth}
\begin{prooftree}
\AxiomC{$\mu(\model{\bone}) = 1$}
\RightLabel{$R^\Ela_\mu$}
\UnaryInfC{$\vdash \bone \Ela \fone$}
\end{prooftree}
\end{minipage}
\end{minipage}

\begin{minipage}{\linewidth}
\begin{minipage}[t]{0.45\linewidth}
\begin{prooftree}
\AxiomC{$\vdash \btwo \Era \fone$}
\AxiomC{$\mu(\model{\btwo}) \ge q$}
\RightLabel{$R^\Era_\BOX$}
\BinaryInfC{$\vdash \bone \Era \BOX^q\fone$}
\end{prooftree}
\end{minipage}
\hfill
\begin{minipage}[t]{0.5\linewidth}
\begin{prooftree}
\AxiomC{$\vdash \btwo \Ela \fone$}
\AxiomC{$\mu(\model{\btwo}) < q$}
\RightLabel{$R^\Ela_\BOX$}
\BinaryInfC{$\vdash \bone \Ela \BOX^q\fone$}
\end{prooftree}
\end{minipage}
\end{minipage}

\begin{minipage}{\linewidth}
\begin{minipage}[t]{0.45\linewidth}
\begin{prooftree}
\AxiomC{$\vdash \btwo \Ela \fone$}
\AxiomC{$\mu(\model{\btwo}) < q$}
\RightLabel{$R^\Era_\DIA$}
\BinaryInfC{$\vdash \bone \Era \DIA^q\fone$}
\end{prooftree}
\end{minipage}
\hfill
\begin{minipage}[t]{0.5\linewidth}
\begin{prooftree}
\AxiomC{$\vdash \btwo \Era \fone$}
\AxiomC{$\mu(\model{\btwo}) \ge q$}
\RightLabel{$R^\Ela_\DIA$}
\BinaryInfC{$\vdash \bone \Ela \DIA^q\fone$}
\end{prooftree}
\end{minipage}
\end{minipage}

}}
\caption{Sequent Calculus $\mathbf{LK}_{\CPLc}$}\label{fig:CPLcproof}
\end{center}
\end{figure}


\section{On the Expressive Power of $\CPLc$}\label{sec1}

Our counting logics are strongly related to probabilistic reasoning
and, indeed,
$\CPLc$ offers a natural model for events 
corresponding to Bernoulli distributions.
In particular, we show how counting formulas can simulate
experiments associated to dyadic distribution.
To do so, we start by introducing auxiliary quantifiers
to express exact probability in a compact way (Section~\ref{sec:3.1}).
Then, we show that our formulas provide a natural formalism
to express (quantitative properties of) events associated
to \emph{dyadic} distributions (Section~\ref{sec:3.2}).
In particular, it is proved that counting formulas can simulate
events associated with any such distribution, but
those related to non-dyadic ones only in an approximate way.
Finally, it is sketched a natural generalization of $\CPLc$ able to simulate events associated to any discrete distribution (Section~\ref{appendix}).

\subsection{Expressing Exact Probability}\label{sec:3.1}
In $\CPLc$, we can easily express that
the probability for a formula to be true
is \emph{precisely} the given one.
In fact, for the sake of readability, 
we introduce auxiliary quantifiers,
$\BBOX^q$ and $\BDIA^q$, 
intuitively meaning that their
argument formula is true with probability (resp.)
strictly greater or smaller than $q$.

\begin{notation}
So-called \emph{white counting quantifiers} are interpreted as follows:

\begin{minipage}{\linewidth}
\begin{minipage}[t]{0.2\linewidth}
$$
\model{\BBOX^q \fone} := \begin{cases}
\twoOm \ &\text{if } \mu_{\mathscr{C}}(\model{\fone}) > q \\
\emptyset \ &\text{otherwise}
\end{cases}
$$
\end{minipage}
\hfill
\begin{minipage}[t]{0.4\linewidth}
$$
\model{\BDIA^q\fone} :=
\begin{cases}
\twoOm \ &\text{if } \mu_{\mathscr{C}}(\model{\fone}) \leq q \\
\emptyset \ &\text{otherwise.}
\end{cases}
$$
\end{minipage}
\end{minipage}
\end{notation}
\noindent
Clearly, these quantifiers do not extend the expressive power
of $\CPLc$, as they are easily definable in terms
of the primitive $\BOX^q$ and $\DIA^q$, see Proposition~\ref{prop:CD}.

\begin{lemma}\label{lemma1}
For every formula of $\CPLc$ $\fone$,
and $q\in \mathbb{Q}_{[0,1]}$,
$$
\mu_{\mathscr{C}}(\model{\fone}) \triangleright q 
\ \ \ \Leftrightarrow 
\ \ \ \mu_{\mathscr{C}}(\model{\neg \fone}) \triangleleft
1 - q,
$$
with $(\triangleright,\triangleleft) \in
\{(\ge,\leq), (\leq, \ge), (>,<), (<,>)\}$.
\end{lemma}
\begin{proof}[Proof Sketch]
Let us consider the case $\leq,\ge$ only.
Since
$\mu_{\mathscr{C}}(\model{\neg \fone})
= \mu_{\mathscr{C}}(\twoOm - \model{\fone})
= 1 - \mu_{\mathscr{C}}(\model{\fone})$,
trivially
$\mu_{\mathscr{C}}(\model{\fone})\ge q
\Leftrightarrow 1 - \mu_{\mathscr{C}}(\model{\fone})
\leq 1-q \Leftrightarrow \mu_{\mathscr{C}}(\model{\neg F})
\leq 1-q$.
\end{proof}

\begin{prop}\label{prop:CD}
For every formula of $\CPLc$
$\fone$, and $q\in \mathbb{Q}_{[0,1]}:$

\begin{minipage}{\linewidth}
\begin{minipage}[t]{0.5\linewidth}
\begin{align*}
\BOX^q \neg \fone &\equiv \BDIA^{1-q}\fone \\
\BOX^q \neg \fone &\equiv \neg \BBOX^{1-q} \fone 
\end{align*}
\end{minipage}
\hfill
\begin{minipage}[t]{0.5\linewidth}
\begin{align*}
\DIA^q\neg \fone &\equiv \BBOX^{1-q} \fone \\
\DIA^{q} \neg \fone &\equiv \neg \BDIA^{1-q}\fone.
\end{align*}
\end{minipage}
\end{minipage}
\end{prop}
\begin{proof}
The proof is based on semantic definition and
Lemma~\ref{lemma1} above:

\begin{minipage}{\linewidth}
\begin{minipage}[t]{0.4\linewidth}
\begin{align*}
\model{\BOX^q\neg \fone} &=
\begin{cases}
\twoOm &\text{if } \mu_{\mathscr{C}}(\model{\neg \fone})
\ge q \\
\emptyset &\text{otherwise}
\end{cases} \\
&\stackrel{L~\ref{lemma1}}{=} \begin{cases}
\twoOm &\text{if } \mu_{\mathscr{C}}(\model{\fone})
\leq 1-q \\
\emptyset &\text{otherwise}
\end{cases} \\
&\ = \model{\BDIA^{1-q}\fone} \\
\\
\model{\BOX^q \neg \fone} & \ =
\model{\neg \DIA^{q}\neg \fone} \\
& \ = \model{\neg \BBOX^{1-q}\fone}
\end{align*}
\end{minipage}
\hfill
\begin{minipage}[t]{0.6\linewidth}
\begin{align*}
\model{\DIA^q\neg \fone}
&= \begin{cases}
\twoOm &\text{if } \mu_{\mathscr{C}}(\model{\neg \fone})
< q \\
\emptyset &\text{otherwise}
\end{cases}\\
&\stackrel{L~\ref{lemma1}}{=}
\begin{cases}
\twoOm &\text{if } \mu_{\mathscr{C}}(\model{\fone})
> 1- q  \\
\emptyset &\text{otherwise}
\end{cases} \\
&\ = \model{\BBOX^{1-q}\fone}. \\
\\
\model{\BOX^q\neg \fone} & \ = \model{\neg \BOX^q\neg \fone}
\\ 
& \ = \model{\neg \BDIA^{1-q}\fone}.
\end{align*} 
\end{minipage}
\end{minipage}

\end{proof}
\noindent
Nevertheless, via $\BBOX^q$ and $\BDIA^q$, we can
express exact probability \emph{in a compact way}.
For example,

\begin{ex}\label{ex1}
We formalize that the formula $\fone=\mathbf{1} \wedge
\mathbf{2}$ is true with probability $\frac{1}{4}$
as:
$$
\fone_{ex} = \BOX^{1/4}(\mathbf{1}\wedge \mathbf{2}) \wedge
\BDIA^{1/4}(\mathbf{1} \wedge \mathbf{2}).
$$
\end{ex}
\noindent
We can even extend $\mathbf{LK}_{\CPLc}$
with the derivable rules for $\BBOX$ and $\BDIA$
illustrated in Figure~\ref{fig:whiteRules}.

\begin{figure}[h!]
\begin{center}
\framebox{
\parbox[t][2cm]{11.5cm}{
\begin{minipage}{\linewidth}
\begin{minipage}[t]{0.45\linewidth}
\begin{prooftree}
\AxiomC{$\vdash \btwo \Era \fone$}
\AxiomC{$\mu(\model{\btwo}) > q$}
\RightLabel{$R^\Era_{\BBOX}$}
\BinaryInfC{$\vdash \bone \Era \BBOX^q \fone$}
\end{prooftree}
\end{minipage}
\hfill
\begin{minipage}[t]{0.5\linewidth}
\begin{prooftree}
\AxiomC{$\vdash \btwo \Ela \fone$}
\AxiomC{$\mu(\model{\btwo}) \leq q$}
\RightLabel{$R^\Ela_{\BBOX}$}
\BinaryInfC{$\vdash \bone \Ela \BBOX^q\fone$}
\end{prooftree}
\end{minipage}
\end{minipage}

\begin{minipage}{\linewidth}
\begin{minipage}[t]{0.45\linewidth}
\begin{prooftree}
\AxiomC{$\vdash \btwo \Era \fone$}
\AxiomC{$\mu(\model{\btwo}) \leq q$}
\RightLabel{$R^\Era_{\BDIA}$}
\BinaryInfC{$\vdash \bone \Era \BDIA^q \fone$}
\end{prooftree}
\end{minipage}
\hfill
\begin{minipage}[t]{0.5\linewidth}
\begin{prooftree}
\AxiomC{$\vdash \btwo \Ela \fone$}
\AxiomC{$\mu(\model{\btwo}) > q$}
\RightLabel{$R^\Ela_{\BDIA}$}
\BinaryInfC{$\vdash \bone \Ela \BDIA^q\fone$}
\end{prooftree}
\end{minipage}
\end{minipage}
}}
\caption{Rules for $\BBOX$ and $\BDIA$}\label{fig:whiteRules}
\end{center}
\end{figure}

\subsection{Simulating Dyadic Distributions}\label{sec:3.2}
It is natural to see atomic formulas of $\CPLc$
as corresponding to infinite sequences of \emph{fair} 
coin tosses and, more in general,
counting formulas as formalizing
experiments associated with \emph{specific} probability distributions. 
For instance, when tossing an unbiased coin twice,
the probability that it returns $\head$ both times
is $\frac{1}{4}$.
This fact is expressed in $\CPLc$ by $F_{ex}$ above, (which
is easily proved valid in our semantics 
and derivable in the corresponding proof system, 
see Appendix~\ref{app:1}).
Generally speaking, counting formulas can simulate events
associated with \emph{any dyadic} distribution, 
but those related to non-dyadic ones only in an approximate
way.\footnote{As we shall see, 
a generalization of $\CPLc$ associated
with a probability space 
$(2^{\mathbb{N}}, \sigma(\mathscr{C}), \mu_{\mathscr{C}})$,
where $\mu_{\mathscr{C}}$ is not necessarily
the measure of i.i.d. sequences are cursorily
presented in Section~\ref{appendix}.
Clearly, these logics can also express events related
to non-dyadic distributions.}

\paragraph{Dyadic Distributions.}
In particular, we formalize atomic sampling from a
Bernoulli distribution of \emph{non-reducible} parameter
$p=\frac{m}{2^n}$ by molecular formulas of $\CPLc$, 
while events are expressed combining such formulas in the usual way.
%
Let us consider a simple example of an experiment
associated with a dyadic distribution.

\begin{ex}\label{ex2}
Let a \emph{biased} coin return
$\head$ only 25$\%$ of the time.
Clearly, in this case, a single toss is not simulated
by an atomic formula of $\CPLc$,
but by a complex one, namely
$(\mathbf{i} \wedge \mathbf{j})$, 
with $i,j\in \Nat$  ``fresh''.
Then, we can even express properties concerning events.
For instance, that the probability for at least one of two subsequent 
biased tosses to return $\head$ is greater than 
$\frac{1}{3}$ is formalized by the (valid) formula:
$$
\fone_{bias} = \BOX^{1/3}\big((\mathbf{1}\wedge \mathbf{2})
\vee (\mathbf{3}\wedge \mathbf{4})\big).
$$
\end{ex}
\noindent
Specifically, we prove that formulas of $\CPLc$
are interpreted as events associated with 
\emph{dyadic distribution} relying on
the notion of cylinder of rank $n$ (and on Corollary~\ref{cor:CylDyad}).
We start by showing that any counting formula (as finite)
is interpreted as a cylinder of a proper rank.
%

\begin{lemma}
For any formula of $\CPLc$ $\fone$,
there is a cylinder of rank $k$ $\cyl_K$ 
such that $\model{\fone} = \cyl_K$.
\end{lemma}
\begin{proof}
The proof is by induction on the structure of $\fone$:
\begin{itemize}
\itemsep0em

\item[$*$] $\fone=\mathbf{i}$ for some $i\in \Nat$.
Then, by Definition~\ref{def:semantics}, $\model{\mathbf{i}}
=Cyl(i)$, which is a thin cylinder.

\item[$*$] $\fone = \neg \ftwo$. By IH, there is a $k\in\Nat$
and a cylinder of rank $k$  $\cyl_K$,
such that $\model{\ftwo}=\cyl_K$.
Let $K'= \Bool^k \setminus K$.
Then, $\model{\neg \ftwo} = \twoOm - \model{\ftwo}
= \twoOm - \cyl_K = \cyl_{K'}$ is clearly a cylinder of rank $k$
as well.

\item[$*$] $\fone$ = $\ftwo \wedge \fthree$.
By IH, there are exist $k_1, k_2\in\Nat$
and cylinders of rank $k_1, k_2$,
such that (resp.) $\model{\ftwo} = \cyl_{K_1}$
and $\model{\fthree}= \cyl_{K_2}$.
Then,
if $k_1=k_2$, $\model{\fone} = \model{\ftwo} \cap \model{\fthree}
= \cyl_{K_1} \cap \cyl_{K_2} = \cyl_{K_1\cap K_2}$
which is a cylinder of rank $k_1$ as well.
Otherwise, assume $k_1> k_2$ (the case $k_2>k_1$ is equivalent). Let $K_2'$ consists of the sequences 
$(u_1,\dots, u_{k_1})$ in $\Bool^{k_1}$ such that the truncated 
sequence $(u_1,\dots, u_{k_2})$ is in $K_2$.
Then, we obtain an alternative, but equivalent 
representation for $\cyl_{K_2}$:
$\cyl_{K_2}' = \{\omega : (\omega(1), \dots \omega(k))
\in K_2'\}$.
So, $\model{\fone}= \model{\ftwo} \cap \model{\fthree}
= \cyl_{K_1} \cap \cyl_{K_2}' = \cyl_{K_1 \cap K_2'}$,
which is a cylinder of rank $k_1$.

\item[$*$] $\fone$ = $\ftwo \wedge \fthree$ is similar to the
case above. 

\item[$*$] $\fone=\BOX^q \ftwo$. Then, by
Definition~\ref{def:semantics}, 
either $\model{\fone}=\twoOm$ or $\model{\fone} = \emptyset$
which are both cylinder of rank (resp.) 0 and $k$. 

\item[$*$] $\fone=\DIA^q \ftwo$. Equivalent to the case above.

\end{itemize}
\end{proof}

\noindent
Then, since by Corollary~\ref{cor:CylDyad},
for any $n$ the measure of a cylinder of rank $n$ is dyadic,
we conclude that for any formula of $\CPLc$
its measure is dyadic as well.\footnote{A ``syntactic proof'' of this result is obtained as a corollary of Lemmas presented in Section~\ref{sec:measuring}.}

\begin{lemma}\label{lemma:nonDyad}
For any formula of $\CPLc$ $\fone$, there are
$n,m\in \Nat$, such that 
$\mu_{\mathscr{C}}(\model{\fone}) =\frac{m}{2^n}$.
\end{lemma}
\begin{proof}
By putting Corollary~\ref{cor:CylDyad} and Lemma~\ref{lemma:nonDyad} together.
\end{proof}


\paragraph{Discrete Distributions with $\sharp X=2^n$}
In the same way, we can (quantitatively) simulate
any discrete distribution with $\sharp X=2^n$,
(which are possibly non-Bernoulli ones).
Of course,
any information concerning the nature of variables
involved in the experiment is lost,
but the quantitative aspects, that is the probability of events,
are all preserved through the formalization.
Let us consider an example of such simulation,
which, as we shall see, passes through 
(somehow arbitrary) ``conventions''.

\begin{ex}
We have four cards, 
each representing a lynx from an existing species,
i.e.~\emph{lynx canadensis},
\emph{lynx lynx}, \emph{lynx pardinus}
and \emph{lynx rufus}.
We can express that, by randomly choosing a card, the
probability not to catch the \emph{lynx lynx} is $\frac{3}{4}$,
for example, by the following (valid) formula:
$$
\fone_{nll} = \BOX^{3/4}\big(\neg(\mathbf{1} \wedge \mathbf{2})\big)
\wedge
\BDIA^{3/4}\big(\neg(\mathbf{1}\wedge \mathbf{2})\big),
$$
whereas the probability that the card represents
a European lynx, that is either the \emph{lynx lynx}
or the \emph{lynx pardinus}, is $\frac{1}{2}$
is formalized by the (valid) formula:
\begin{align*}
F_{el} &= \BOX^{1/2}\big((\mathbf{1}\wedge\mathbf{2})
\vee (\mathbf{1} \wedge \neg \mathbf{2})\big)
\wedge
\BDIA^{1/2}\big((\mathbf{1}\wedge\mathbf{2})
\vee (\mathbf{1} \wedge \neg \mathbf{2})\big) \\
&= \BOX^{1/2}\mathbf{1} \wedge
\BDIA^{1/2}\mathbf{1}.
\end{align*}
If one repeats twice the (random) picking of 
a card from the four available ones,
then the probability that a \emph{lynx rufus}
is picked both times is $\frac{1}{16}$ as simulated
by the formula below:
$$
F_{lr2} = \BOX^{1/16}\big((\neg \mathbf{1} \wedge \neg \mathbf{2})
\wedge (\neg \mathbf{3} \wedge \neg \mathbf{4})\big)
\wedge
\BDIA^{1/16}\big((\neg \mathbf{1} \wedge \neg \mathbf{2})
\wedge (\neg \mathbf{3} \wedge \neg \mathbf{4})\big).
$$
\end{ex}

\paragraph{Non-Dyadic Bernoulli Distributions.}
Something different happens when considering experiments related
to non-dyadic distributions.
Indeed, by Lemma~\ref{lemma:nonDyad},
formulas of $\CPLc$
cannot express these events as they are always
interpreted as sets of dyadic measure.
Nevertheless, we can somehow simulate events
associated to non-dyadic distributions
\emph{in an approximate way}.
%

\begin{ex}\label{ex:ndyad}
Let us consider a biased coin returning $\head$ 
only $\frac{1}{3}$ of the time.
We cannot simulate this event in $\CPLc$.
We can however ``approximate'' it with $n=2m$
variables of $\CPLc$ in the following sense.
For instance, if $m=2$ we can \emph{down-approximate} 
a single toss of the biased coin as follows:

$$
\fone_{ndy}= 
(\mathbf{1} \wedge \mathbf{2})
\vee 
\big((\neg \mathbf{1} \wedge \mathbf{2})
\wedge (\mathbf{3} \wedge \mathbf{4})\big).
$$
Notably, disjuncts are mutually contradictory
so $\mu_{\mathscr{C}}(\model{\fone_{ndy}}) = 
\frac{1}{2^2} + \frac{1}{2^4}
= \frac{5}{16}$.

If $m=3$, (down-)approximation is obtained as follows,
$$
\fone_{ndy}' = 
(\mathbf{1} \wedge \mathbf{2})
\vee 
\big((\neg \mathbf{1} \wedge \mathbf{2})
\wedge (\mathbf{3} \wedge \mathbf{4})\big)
\vee
\big((\neg \mathbf{1} \wedge \mathbf{2})
\wedge (\neg \mathbf{3} \wedge \mathbf{4})
\wedge (\mathbf{5} \wedge \mathbf{6})\big).
$$
In this case, $\mu_{\mathscr{C}}(\model{\fone_{ndy}}')
= \frac{1}{2^2} + \frac{1}{2^4} +
\frac{1}{2^6} = \frac{21}{64}$.

In general, 
increasing $n$, our formula of $\CPLc$
approximates the desired atomic event in a 
more and more precise way.
\end{ex}
\noindent
Although events associated with non-dyadic
distributions cannot be expressed in $\CPLc$ in a precise
way, when switching to the measure quantified language $\MQPA$, as presented in~\cite{CiE}, this formalization becomes
possible.\footnote{This fact is coherent with our randomized
arithmetization result~\cite[Theorem 3]{CiE}.}
In the following Section~\ref{appendix}, 
we provide an alternative, 
more natural way to generalize $\CPLc$ so to express such events.


\subsection{Generalizing $\CPLc$}\label{appendix}

As seen, the semantics for $\CPLc$ is associated
with a canonical cylinder space 
$\mathscr{P} = (\twoOm,\sigma(\mathscr{C}), 
\mu_{\mathscr{C}})$,
where $\mu_{\mathscr{C}}$ is the standard measure
$\mu_{\mathscr{C}}(Cyl(i))=\frac{1}{2}$
for any $i\in\Nat$, 
intuitively corresponding to tossing \emph{fair}
coins~\cite{Billingsley}.
We can generalize this framework in a natural
way, so to allow
the cylinder measure  to possibly express experiments
associated to distributions
other than dyadic ones.
So, we define extended $\CPLc^*$,
this time associated
with $\mathscr{P}^\star 
= (2^\mathbb{N}, \sigma(\mathscr{C}), \mu^\star)$,
where $\mu^\star$ can be \emph{any} (well-defined) probability
measure over $\sigma(\mathscr{C})$.
  In particular, the grammar and semantics 
  for $\CPLc^\star$ is as in
Definition~\ref{def:semantics},
except for quantifiers, which are now 
defined using the probability measure,
$\mu^\star$.

\begin{defn}[Grammar and Semantics of $\CPLc^\star$]
Extended formulas are defined by substituting standard
counting quantifiers with $\BOX^q_{\mu^*}$ and $\DIA^{q}_{\mu^\star}$,
the interpretation of which is generalized as basing on
$\mathscr{P}^\star$:

\begin{minipage}{\linewidth}
\begin{minipage}[t]{0.1\linewidth}
$$
\model{\BOX^q_{\mu^\star}\fone} =
\begin{cases}
\twoOm \ &\text{if } \mu^\star(\llbracket \fone\rrbracket) \ge q \\
\emptyset \ &\text{otherwise}
\end{cases}
$$
\end{minipage}
\hfill
\begin{minipage}[t]{0.39\linewidth}
$$
\model{\DIA^q_{\mu^\star}\fone} = \begin{cases}
\twoOm \ &\text{if } \mu^\star(\model{\fone}) < q \\
\emptyset \ &\text{otherwise.}
\end{cases}
$$
\end{minipage}
\end{minipage}
\end{defn}
\noindent
Then, $\CPLc$-formulas $\BOX^q\fone$
and $\DIA^q\fone$ become special cases of 
the extended
ones, (resp.) $\BOX^q_{\mu^\star}\fone$ and $\DIA^q_{\mu^\star}
\fone$ where $\mu^\star=\mu_{\mathscr{C}}$.
On the other hand, in $\CPLc^\star$
we can simulate experiments corresponding
to tossing (arbitrarily) \emph{biased} coins 
in a very simple way.
Let us consider the experiment of Example~\ref{ex:ndyad}
once again.

\begin{ex}
Let a biased coin again return $\head$ only $\frac{1}{3}$
of the time.
We consider a specific $\mu*$ such that 
$\mu^* (Cyl(i)) = \frac{1}{3}$ for any $i\in\Nat$.
Then, we can express that the probability for subsequent
tosses to be successful is greater than $\frac{1}{9}$
by the formula of $\CPL_0^\star$:
$
\fone_\star = \BOX^{1/9}_{\mu^*}(\mathbf{1} \wedge \mathbf{2}).
$
Furthermore, since $\mu^*(\model{\mathbf{1}\wedge \mathbf{2}}) = \frac{1}{3}\cdot \frac{1}{3}$,
the formula is valid, i.e. $\model{\BOX^{1/9}_{\mu^*}(\mathbf{1}\wedge \mathbf{2})}=\twoOm$.
\end{ex}

Observe that it is easy to define a sound and complete 
proof system for $\CPL_0^\star$.
Indeed, no substantial change is required with respect
to the calculus $\mathbf{LK}_{\CPLc}$ of~\cite[Sec. 2.2]{ICTCS}.
Indeed, only semantic conditions are related
to probability measure, so one only needs to generalize
$\mu$-rules substituting $\mu_{\mathscr{C}}$
with $\mu^\star$, in the natural way illustrated
by Figure~\ref{fig:CPLstar}.

\begin{figure}[h!]
\begin{center}
\framebox{
\parbox[t][1.2cm]{10cm}{

\begin{minipage}{\linewidth}
\begin{minipage}[t]{0.4\linewidth}
\begin{prooftree}
\AxiomC{$\mu^\star(\model{\bone}) = 0$}
\RightLabel{$R^\Era_{\mu^\star}$}
\UnaryInfC{$\vdash \bone \Era \fone$}
\end{prooftree}
\end{minipage}
\hfill
\begin{minipage}[t]{0.5\linewidth}
\begin{prooftree}
\AxiomC{$\vdash \bone \Era \fone$}
\AxiomC{$\mu^\star(\model{\btwo}) \ge q$}
\RightLabel{$R^\Era_\BOX$}
\BinaryInfC{$\vdash \bone \Era \BOX^q_{\mu^\star} \fone$}
\end{prooftree}
\end{minipage}
\end{minipage}
}}
\caption{Examples of Rules in $\mathbf{LK}_{\CPLc^\star}$}\label{fig:CPLstar}
\end{center}
\end{figure}


\section{Measuring Formulas of $\CPLc$}\label{sec:measuring}

In~\cite{ICTCS},
the validity of counting formulas is decided
accessing an oracle for $\shSat$,
i.e.~counting the satisfying models of Boolean formulas.
Here, we provide an effective procedure to measure
formulas of $\CPLc$,
\emph{without appealing for an external source},
so somehow making explicit the task which, in the proof system $\mathbf{LK}_{\CPLc}$ is done by the oracle.
In this way, there is no need for so-called external hypotheses and a purely syntactical (though labelled) calculus can  be defined.
In our opinion, this result also makes the comparison with
other probability logics, such as~\cite{FHM},
more clear and help clarifying the nature of counting quantifiers.
Furthermore, we hope this is also the first step to shed new
lights on the study of the complexity of deciding
formulas of $\CPLc$.

Our proof is based on some crucial steps. 
Given a counting formula,
we start by considering its (inner) 
quantified formulas.
To calculate their measure, we pass through a special form,
the measure of which can be computed in a straightforward
way (Lemma~\ref{lemma:measure}),
and prove that formulas of $\CPLc$
\emph{without quantifiers} can be converted
into such measurable form (Lemma~\ref{lemma:conv}).
Notably, this procedure is effective, but not necessarily
``feasible'' as requiring argument formulas to be in disjunctive normal form (DNF, for short).
Finally, nested quantifications can be taken 
into account.
It is easy to see that by measuring argument formulas,
we can simply substitute the corresponding
(non-nested) counting-quantified expressions is either
valid or invalid (see~\cite[p. 7]{ICTCS}).

\subsection{Preliminaries}
On the logical side, 
we enrich the language of $\CPLc$
by the two standard symbols $\top$ and $\bot$,
to be interpreted as predictable, i.e.~$\model{\top}=\twoOm$ and $\model{\bot}=\emptyset$.
\begin{notation}
We use $L_1,L_2,\dots$ to denote literals,
that is either a counting variable or its negation.
Given a literal $L_i$,
$$
\overline{L_i} = \begin{cases}
\neg \mathbf{k} \ \ \ &\text{if }  L_i = \mathbf{k} \\
\mathbf{k} \ \ \ &\text{otherwise.}
\end{cases}
$$
\end{notation}

\subsection{Polite Normal Form}
For simplicity, before defining \emph{measurable normal form},
we introduce the auxiliary, polite form.
%
%
%
%
%
\noindent
 Intuitively, a conjunction of literals is in polite form if (is
in $\{\bot,\top\}$ or) 
each variable occurs in it at most once,
i.e.~repetitions are removed and it is not possible for
any literal to appear in the conjunction both in atomic
  $L_i$, and in negated form, $\overline{L_i}$.
 Observe that the measure
of this conjunction can be easily computed.
   Trivially, when the formula is $\top$, 
   its measure is 1 and when is
   $\bot$, its measure is 0.
   Otherwise, since literals correspond
 to mutually \emph{independent} events
 (each of measure $\frac{1}{2}$),
 the measure of the polite conjunction of $n$
 formulas is simply $\frac{1}{2^n}$.
  On the other hand, a formula in DNF
is in \emph{disjunctive polite form} when it is either
$\top$ or $\bot$ or 
each of its disjunct is in polite form, without
being in $\{\bot,\top\}$.
Formally,

\begin{defn}[Polite Form]\label{def:PNF}
A formula of $\CPLc$, which is a conjunction of literals
$C=\bigwedge_{i\in\{1,\dots,n\}}L_i$,
is in \emph{conjunctive polite form} (CPF,
for short) if either $C\in\{\bot,\top\}$
or for any $k\neq k'\in\{1,\dots,n\}$,
$L_k\neq L_{k'}$ and
$L_k \neq \overline{L_{k'}}$.
A formula of $\CPLc$ in DNF
$D=\bigvee_{j\in\{1,\dots,m\}} C_j$,
is in \emph{disjunctive polite form} (DPF,
for short) if either $D\in\{\bot,\top\}$
or for each $k\in\{1,\dots,m\}$,
$C_k$ is in CPF and $C_k\not\in\{\bot,\top\}$.
\end{defn}

Then, we prove that every formula of $\CPLc$
(without quantifiers)
can be converted in DPF.
In particular, as for standard $\PL$,
we know that for any counting formula 
\emph{without quantifiers},
there is an equivalent formula of $\CPLc$ in DNF.
Then, by Lemma~\ref{lemma:DPF} below, 
it is established that
for every DNF-formula,
there is an equivalent formula of $\CPLc$
in DPF.\footnote{Actually, the proof of Lemma~\ref{lemma:DPF},
is so defined that the DPF-formula $D^*$ does not contain
repetitions of disjuncts.}

\begin{lemma}\label{lemma:DPF}
Given a formula of $\CPLc$ in DNF $D$,
there is a $D^*$ such that $D^*$ is in DPF and
$D\equiv D^*$.
\end{lemma}

\begin{proof}
Let $D=\bigvee_{i\in\{1,\dots,n\}}C_i$ be in DNF.
For any $C_i = \bigwedge_{j\in\{1,\dots,m\}}L_j$,
with $i\in\{1,\dots,n\}$,
we define $C_i^*$ applying the transformations
below:
\begin{itemize}
\itemsep0em
\item[$i.$] if $C_i=\top$, then $C_i^*=\top$.
\item[$ii.$] otherwise, consider each $j\in\{1,\dots,m\}$,
starting from $j=1$:
\begin{itemize}
\itemsep0em
\item[$a.$] if $L_j=\bot$, then $C^*_i=\bot$.
\item[$b.$] if $L_j=\top$, then $L_j$ is removed
and $j+1$ is considered.
\item[$c.$] if $L_j\not\in\{\bot,\top\}$, 
we consider each pedex $k\neq j\in\{1,\dots, m\}$, starting from the first possible one:
\begin{itemize}
\itemsep0em
\item if $L_j=L_k$, then $L_k$ is removed and the subsequent pedex (different from $j$ and $k$) is considered.
\item if $L_j=\overline{L_k}$, then $C_i^*=\bot$.
\item otherwise, $L_j$ is left unchanged and $k+1$
is considered.
\end{itemize}
\end{itemize}
\end{itemize}
It is clear that $C_i\equiv C_i^*$.
Now we consider $D'=\bigvee_{i\in\{1,\dots,n'\}}C_i^*$
and define $D^*$ applying the following transformations:
\begin{itemize}
\itemsep0em
\item[$i.$]
if $C_i^*=\bot$ for any $i\in\{1,\dots,n'\}$,
then $D^*=\bot$.

\item[$ii.$] otherwise, we consider each $i\in\{1,\dots,n'\}$,
starting from $i=1$:
\begin{itemize}
\itemsep0em
\item[$a.$] if $C_i=\top$, then $D^* = \top$.

\item[$b.$] if $C_i=\bot$, then $C_i$ is removed and
$i+1$ is considered.

\item[$c.$] if $C_i\not\in\{\top,\bot\}$, we consider
each pedex starting from the first $k'\neq i\in\{1,\dots,n'\}$:

\begin{itemize}
\itemsep0em
\item if $C_i$ and $C_{k'}$
contain exactly the same literals,
then $C_{k'}$ is removed and the subsequent pedex 
($\neq k',i$) is considered.

\item Otherwise, $C_i$ is (maybe temporarily) left unchanged
and the subsequent pedex ($\neq k',i$) is considered.
\end{itemize}
\end{itemize}
\end{itemize}
Again, it is clear that $D\equiv D^*$.
\end{proof}
\noindent
Observe that for any formula $\fone$
in DPF,
either $\fone\in\{\bot,\top\}$
or no instance of $\bot,\top$ occurs in it.
Furthermore, as anticipated,
it is easy to measure the probability of
a formula in CPF.

\begin{prop}\label{prop:CPF}
Given a formula $C$ in CPF:
\begin{itemize}
\itemsep0em
\item[$i.$] if $C=\top$, then $\mu_{\mathscr{C}}(\model{C}) = 1$.

\item[$ii.$] if $C=\bot$,
then $\mu_{\mathscr{C}}(\model{C})=0$.

\item[$iii.$] Otherwise, $C=\bigwedge_{i\in\{1,\dots,n\}}
L_i$, and 
$\mu_{\mathscr{C}}(\model{C})= \frac{1}{2^n}$.
\end{itemize}
\end{prop}

\begin{proof}
Case $i.$ and $ii.$ are trivial consequences of
Definition~\ref{def:semantics} and basic
measure theory:
$\mu_{\mathscr{C}}(\model{\top})= \mu_\mathscr{C}(\twoOm)
=1$
and 
$\mu_{\mathscr{C}}(\model{\top})= \mu_\mathscr{C}(\emptyset)
=0$.
Case $iii.$ relies on Definition~\ref{def:PNF},
that is $C$ does not contain $\bot$ or $\top$ (or contradictions) 
or repetitions.
Thus, by semantic definition, its literals have to be
interpreted as \emph{independent} events the measure
of which is known, i.e.~$\model{L_i}$'s for $i\in\{1,\dots,n\}$ are
\emph{independent} events.
So, as seen, for basic measure theory,
$\mu_{\mathscr{C}}\big(\big\llbracket\bigwedge_{i\in\{1,\dots,n\}}L_i\big\rrbracket\big)
= \mu_{\mathscr{C}}\big(\bigcap_{i\in\{1,\dots,n\}} \model{L_i}\big)
= \frac{1}{2^n}$.
\end{proof}

\subsection{Measurable Normal Form}

Now, we introduce a special form, such that formulas
in this form can be 
 ``measured'' in a straightforward way.
 We start by considering the (logical) notion of 
 \emph{contradictory pair}.
 Two formulas in CPF are mutually contradictory
 if their conjunction is an invalid formula.

\begin{defn}[Contradictory Pair]
Two formulas of $\CPLc$ in CPF, $C_i=\bigwedge_{j\in\{1,\dots,n\}} L_j$
and 
$C_i'=\bigwedge_{k\in\{1,\dots,m\}}L_k$
are said to be \emph{mutually contradictory}
when there exist $j\in\{1,\dots,n\},k\in\{1,\dots,m\}$
such that $L_j=\overline{L_k}$ (or $L_k=\overline{L_j}$).
\end{defn}
\noindent
Clearly, contradictory formulas are interpreted
as \emph{disjoint} events
(and, as seen, the measure of the union
of two disjoint events is the sum of the
measure of each event).
So by Definition~\ref{def:semantics}
plus basic probability theory,
the measure of the disjunction of two contradictory 
formulas ($\not\in\{\bot,\top\}$) in CPF is the sum of the measure of each disjunct
(which, being themselves in CPF,
are easily measurable as well by Proposition~\ref{prop:CPF}).
By generalizing this intuition
we obtain the definition below.

\begin{defn}[Measurable Normal Form]\label{def:MNF}
A formula of $\CPLc$ $\fone=\bigvee_{i\in\{1,\dots,n\}}
C_i$ is in \emph{measurable normal form}
(MNF, for short) 
if either $\fone\in \{\bot, \top\}$ or
$\fone$ is in DPF and for each
$j\neq k\in\{1,\dots,n\}$,
$C_j$ and $C_k$ are mutually contradictory.
\end{defn}
\noindent
Observe that, when a formula is in MNF,
as disjuncts are mutually contradictory,  
no disjunct is repeated in it and no
disjunct can be a sub-formula of another.\footnote{Further details on the notion of sub-formula
are presented in Section~\ref{sec:convMNF}.}

As seen, by Definition~\ref{def:MNF}
a formula of MNF contains disjunct corresponding 
to mutually disjoint events, 
so Lemma~\ref{lemma:measure} naturally follows.

\begin{lemma}\label{lemma:measure}
Given a formula of $\CPLc$ in MNF
$\fone=\bigvee_{i\in\{1,\dots,n\}}C_i$:
\begin{itemize}
\itemsep0em
\item[i.] if $\fone=\bot$, then $\mu_{\mathscr{C}}(\model{\fone})=0$,
\item[ii.] if $\fone=\top$, then $\mu_{\mathscr{C}}(\model{\fone})=1$,
\item[iii.] otherwise, $\mu_{\mathscr{C}}(\model{\fone})=
\sum_{i\in\{1,\dots,n\}} \mu_{\mathscr{C}}(\model{C_i}).$
\end{itemize}
\end{lemma}
\begin{proof}[Proof Sketch]
Case $i.$ and $ii.$ hold by Definition~\ref{def:semantics}
and basic measure theory:
$\mu_{\mathscr{C}}(\model{\bot})=\mu_{\mathscr{C}}(\emptyset)=0$
and 
$\mu_{\mathscr{C}}(\model{\top})=
\mu_{\mathscr{C}}(\twoOm)=1$.
Case $iii.$ is proved relying on the Definition~\ref{def:MNF}:
for any $j\neq k\in\{1,\dots,n\}$, $(C_j,C_k)$ is a contradictory
pair.
Clearly, $\model{C_j}\cap \model{C_k}=\emptyset$
for any $j$ and $k$.
So,
$\mu_{\mathscr{C}}\big(\big\llbracket \bigvee_{i\in\{1,\dots,n\}}
C_i\big\rrbracket\big) =
\mu_{\mathscr{C}}\big(\bigcup_{i\in\{1,\dots,n\}}\model{C_i}\big)
=\sum_{i\in\{1,\dots,n\}}\mu_{\mathscr{C}}(\model{C_i}).
$ 
\end{proof}
\noindent
As said above, each disjunct is in CPF,
so its measure is easily computable as well.
%

\begin{cor}\label{cor1}
Given a formula of $\CPLc$
in MNF 
$\fone=\underbrace{\bigwedge_{i\in\{1,\dots,m_1\}}L_i 
\vee \dots  \vee \bigwedge_{j\in\{1,\dots, m_{n}\}}L_j}_{n \ times}:$
\begin{itemize}
\itemsep0em
\item if $F=\top$, then $\mu_{\mathscr{C}}(\model{\fone})=1$,

\item if $F=\bot$, then $\mu_{\mathscr{C}}(\model{\fone})=0$,

\item otherwise, $\mu_{\mathscr{C}}(\model{\fone})
= \underbrace{1/2^{m_1} + \dots + 1/2^{m_n}}_{n \ times}$.
\end{itemize}
\end{cor}
\begin{proof}
By putting Proposition~\ref{prop:CPF} and Lemma~\ref{lemma:measure} together.
\end{proof}
%
%
%

\subsection{Conversion into MNF}\label{sec:convMNF}

To conclude our proof, we show that
each formula in DPF can actually
be ``converted'' into MNF.
To do so, we notice that two disjuncts
can be mutually related in three ways only:
(1)
one is a sub-formula of the other,
in this case,
the former is simply removed;
(2)
they are a contradictory pair,
so are already in the desired form and
another pair can be considered;
(3)
one of the two disjuncts, say
$C_i$, contains a literal $L_k$,
such that neither $L_k$ or $\overline{L_k}$
occurs in the other disjunct, say $C_j$,
in this case $C_j$ is substituted by
$C_j' = C_j\wedge L_k$ and
$C_j'' = C_j\wedge \overline{L_k}$.
Notice that a formula in CPF,
say $C_i=\bigwedge_{k\in\{1,\dots,n\}} L_k$
is said to be a \emph{sub-formula}
of another formula in CPF 
$C_j=\bigwedge_{k'\in\{1,\dots,m\}}L_{k'}$
when (they are not a contradictory pair and)
 for each $k'\in\{1,\dots,m\}$,
there is a $k\in\{1,\dots,n\}$ such that 
$L_k=L_{k'}$.
For example, the formula
$\mathbf{1} \wedge \mathbf{2}
\wedge \mathbf{3}$
is a sub-formula of $\mathbf{1}\wedge \mathbf{2}$.

\begin{lemma}\label{lemma:conv}
For each formula of $\CPLc$ in DPF
$\fone$, there is a formula in MNF
$\fone^{**}$ such that $\fone\equiv \fone^{**}$.
\end{lemma}

\begin{proof}
Given a formula of $\CPLc$ in DPF
$\fone=\bigvee_{i\in\{1,\dots,n\}}C_i$,
we define a formula $\fone^{**}$
in MNF as follows.
\begin{itemize}
\itemsep0em
\item[$*$] if $\fone\in\{\bot,\top\}$, then $\fone^{**}=
\fone$.

\item[$*$] Otherwise, we consider each $i\in\{1,\dots,n\}$,
starting from $i=1$:
\begin{itemize}

\itemsep0em
\item[$i.$] if there is a $j\neq i\in\{1,\dots, n\}$,
such that $C_i$ is a sub-formula of $C_j$,
then $C_i$ is removed and $i+1$ is considered.

\item[$ii.$] otherwise, we consider each pedexes
$j\neq i\in\{1,\dots,n\}$ starting from the first possible $j$:
\begin{itemize}
\itemsep0em

\item[$a.$] if $C_i$ and $C_j$ are mutually contradictory,
then $j+1$ is considered.

\item[$b.$] otherwise, for $C_i=\bigwedge_{k\in\{1,\dots,l\}}
L_k$
and $C_j=\bigwedge_{k'\in\{1,\dots,l'\}}L_{k'}$,
we consider each $k\in\{1,\dots,l\}$ starting
from $k=1$.
\begin{itemize}
\itemsep0em
\item if there is a $k'\in\{1,\dots,l'\}$
such that $L_k=L_{k'}$, then $L_k$ is left unchanged
and $k+1$ is considered.

\item if there is no $k'\in\{1,\dots,l'\}$
such that $L_k=L_{k'}$, then
we replace $C_i$ with two formulas
$C_i'=C_i \wedge L_{k'}$ and
$C_i'' = C_i \wedge \overline{L_{k'}}$
and apply $a.$-$b.$ to both.
\end{itemize}
We consider each $k'\in\{1,\dots,l'\}$
starting from 1:
\begin{itemize}
\itemsep0em

\item if there is a $k\in\{1,\dots,l\}$ such that
$L_k=L_{k'}$,
then $L_{k'}$ is left unchanged and $k+1$ is
considered.

\item if there is no $k\in\{1,\dots,l\}$
such that $L_k=L_{k'}$, then
we replace $C_j$ with $C_j'=C_j\wedge
L_{k'}$ and 
$C_j'' = C_j\wedge \overline{L_{k'}}$ and 
apply $a.$-$b.$ to both.
\end{itemize}
Then, we consider $j+1$.
\end{itemize}
When $j+1=n$, $i+1$ is considered (until
also $i+1=n$).
\end{itemize}
\end{itemize}
\end{proof}
  So, for any formula of $\CPLc$ \emph{without quantifier}, by Lemma~\ref{lemma:conv}, we can  
``convert'' it into (DNF, then DPF, so) MNF
and this concludes our proof.
  Indeed, by Lemma~\ref{lemma:measure}, given such
 MNF-formula there is a procedure to measure it.
  As anticipated, this result can be extended to any formulas of
 $\CPLc$, as quantified expressions are simply substituted
 by either $\bot$ or $\top$ (depending, of course, 
 on the quantification involved and
 on the fact that the measure of the argument formula
 is greater or strictly smaller than the probability index).
Then, the procedure is repeated as long as required.

\begin{ex}
Let us consider the formula $\fone=\mathbf{1} \vee
(\mathbf{1} \wedge \mathbf{3})$.
Its MNF is $\fone'=\mathbf{1}$.
Then, $\mu_{\mathscr{C}}(\fone)= \frac{1}{2}$.
\\
As a second example, let $\ftwo=\mathbf{1}\vee
(\neg\mathbf{1}\wedge\mathbf{2})$.
This formula is already in MNF and 
$\mu_{\mathscr{C}}(\model{\ftwo})=\frac{1}{2}
+ \frac{1}{4}=\frac{3}{4}$.
\\
Finally, let $\fthree=\mathbf{1}\vee (
\neg \mathbf{1} \wedge \mathbf{2})
\vee (\neg \mathbf{2} \wedge \mathbf{3})$.
Its MNF is
$
\fthree' = (\mathbf{1} \wedge \mathbf{2} \wedge \mathbf{3})
\vee (\mathbf{1} \wedge \mathbf{2} \wedge \neg\mathbf{3})
\vee (\mathbf{1} \wedge \neg \mathbf{2} \wedge \mathbf{3})
\vee (\mathbf{1} \wedge \neg \mathbf{2} \wedge \neg \mathbf{3}) \vee 
(\neg \mathbf{1} \wedge \mathbf{2} \wedge \mathbf{3})
\vee (\neg \mathbf{1} \wedge \mathbf{2} \wedge
\neg \mathbf{3})
\vee (\neg \mathbf{1} \wedge \neg \mathbf{2} \wedge \mathbf{3}).
$
Thus, $\mu_{\mathscr{C}}(\model{\fthree})=\frac{7}{8}$.
\end{ex}


\bibliographystyle{plain}
\bibliography{bib}

\newpage
\appendix
\normalsize

\small
\section{Appendix}\label{app:1}


\begin{proof}[Proof of Example~\ref{ex1}]
The formula $\BOX^{1/4}(\mathbf{1} \wedge \mathbf{2})
\wedge \BDIA^{1/4} (\mathbf{1} \wedge \mathbf{2})$
is shown valid as follows:
\begin{align*}
\model{F_{ex}} &= 
\model{\BOX^{1/4}(\mathbf{1} \wedge \mathbf{2})}
\cap \model{\BDIA^{1/4}(\mathbf{1} \wedge \mathbf{2})} \\
&= \begin{cases}
\twoOm &\text{if } \mu_{\mathscr{C}}(\model{\mathbf{1}
\wedge \mathbf{2}}) \ge \frac{1}{4} \\
\emptyset &\text{otherwise}
\end{cases}
\cap 
\begin{cases}
\twoOm  &\text{if } \mu_{\mathscr{C}}(\model{\mathbf{1}
\wedge \mathbf{2}}) \leq \frac{1}{4} \\
\emptyset &\text{otherwise}
\end{cases} \\
&= \begin{cases}
\twoOm &\text{if } \mu_{\mathscr{C}}(\model{\mathbf{1}} 
\cap \model{\mathbf{2}}) \ge \frac{1}{4} \\
\emptyset &\text{otherwise}
\end{cases}
\cap
\begin{cases}
\twoOm &\text{if } \mu_{\mathscr{C}}(\model{\mathbf{1}}
\cap \model{\mathbf{2}}) \leq \frac{1}{4} \\
\emptyset &\text{otherwise}
\end{cases} \\
&= \begin{cases}
\twoOm &\text{if } \mu_{\mathscr{C}}(Cyl(1)
\cap Cyl(2)) \ge \frac{1}{4} \\
\emptyset &\text{otherwise}
\end{cases} 
\cap 
\begin{cases}
\twoOm &\text{if } \mu_{\mathscr{C}}(Cyl(1) \cap
Cyl(2)) \leq \frac{1}{4} \\
\emptyset &\text{otherwise}
\end{cases} \\
&= \twoOm \cap \twoOm \\
&= \twoOm.
\end{align*} 
The formula
is proved derivable in $\mathbf{LK}_{\CPLc}$:

\begin{prooftree}
\AxiomC{$\mathcal{D}$}
\AxiomC{$\mu(\model{\bvar_1 \wedge \bvar_2})
\ge \frac{1}{4}$}
\RightLabel{$R^\Era_\BOX$}
\BinaryInfC{$\vdash \top \Era \BOX^{1/4} (\mathbf{1} \wedge
\mathbf{2})$}
\AxiomC{$\mathcal{D}$}
\AxiomC{$\mu(\model{\bvar_1 \wedge \bvar_2})<\frac{1}{4}$}
\RightLabel{$R^\Era_\BDIA$}
\BinaryInfC{$\vdash\top \Era \BDIA^{1/4}(\mathbf{1} \wedge
\mathbf{2})$}
\RightLabel{$R^\Era_\wedge$}
\BinaryInfC{$\vdash \bvar_1 \wedge \bvar_2 \Era 
\mathbf{1} \wedge \mathbf{2}$}
\end{prooftree}
where $\mathcal{D}$ is defined as,
\begin{prooftree}
\AxiomC{$\bvar_1 \vDash \bvar_1$}
\RightLabel{$Ax1$}
\UnaryInfC{$\vdash \bvar_1 \Era \mathbf{1}$}
\RightLabel{$R^\Era_\cup$}
\UnaryInfC{$\vdash \bvar_1 \wedge \bvar_2 \Era \mathbf{1}$}
\AxiomC{$\bvar_2 \vDash \bvar_2$}
\RightLabel{$Ax1$}
\UnaryInfC{$\vdash \bvar_2 \Era \mathbf{2}$}
\RightLabel{$R_\cup^\Era$}
\UnaryInfC{$\vdash \bvar_1 \wedge \bvar_2 \Era
\mathbf{2}$}
\RightLabel{$R1^\Era_\wedge$}
\BinaryInfC{$\vdash \bvar_1 \wedge \bvar_2 \Era \mathbf{1}
\wedge \mathbf{2}$}
\end{prooftree}
\end{proof}

\begin{proof}[Proof of Example~\ref{ex2}]
The formula $\fone_{bias}$ of Example~\ref{ex2}
can be easily proved valid.
\begin{align*}
\model{\fone_{bias}} &=
\begin{cases}
\twoOm \ &\text{if } \mu_{\mathscr{C}}\big(\model{(\mathbf{1}
\wedge \mathbf{2}) \vee
(\mathbf{3} \wedge \mathbf{4})}\big) \ge \frac{1}{3} \\
\emptyset \ &\text{otherwise}
\end{cases} \\
&= \begin{cases}
\twoOm \ &\text{if } \mu_{\mathscr{C}}\big(\model{\mathbf{1} \wedge
\mathbf{2}} \cup
\model{\mathbf{3} \wedge \mathbf{4}}\big) \ge \frac{1}{3} \\
\emptyset \ &\text{otherwise} 
\end{cases} \\
&= \begin{cases}
\twoOm \ &\text{if } \mu_{\mathscr{C}}\big((\model{\mathbf{1}}
\cap \model{\mathbf{2}}) \cup
(\model{\mathbf{3}} \cap \model{\mathbf{4}})\big) \ge \frac{1}{3} \\
\emptyset \ &\text{otherwise}
\end{cases} \\
&=
\begin{cases}
\twoOm \ &\text{if } \mu_{\mathbf{C}}\big( (Cyl(1) \cap
Cyl(2)) \cup (Cyl(3) \cap Cyl(4)\big) \ge \frac{1}{3} \\
\emptyset \ &\text{otherwise}
\end{cases} \\
&= \begin{cases}
\twoOm \ &\text{if } \frac{1}{2} \ge \frac{1}{3} \\
\emptyset \ &\text{otherwise}
\end{cases}
\\
&= \twoOm.
\end{align*}
Let $\mathcal{D}$ be defined as in Example~\ref{ex1}.
Then, the derivation of $\fone_{bias}$ in $\mathbf{LK}_{\CPLc}$
is as follows:

\begin{prooftree}
\AxiomC{$\mathcal{D}$}
\RightLabel{$R^\Era_\vee$}
\UnaryInfC{$\vdash \bvar_1 \wedge \bvar_2 \Era (\mathbf{1}
\wedge \mathbf{2}) \vee
(\mathbf{3} \wedge \mathbf{4})$}
\AxiomC{$\bvar_3 \vDash \bvar_3$}
\RightLabel{$Ax1$}
\UnaryInfC{$\vdash \bvar_3 \Era \mathbf{3}$}
\RightLabel{$R1^\Era_\cup$}
\UnaryInfC{$\vdash \bvar_3 \wedge \bvar_4 \Era \mathbf{3}$}
\AxiomC{$\bvar_4 \vDash \bvar_4$}
\RightLabel{$Ax1$}
\UnaryInfC{$\vdash \bvar_4 \Era \mathbf{4}$}
\RightLabel{$R^\Era_\cup$}
\UnaryInfC{$\vdash \bvar_3 \wedge \bvar_4 \Era \mathbf{4}$}
\RightLabel{$R_\wedge^\Era$}
\BinaryInfC{$\vdash \bvar_3 \wedge \bvar_4 \Era 
\mathbf{3} \wedge \mathbf{4}$}
\RightLabel{$R^\Era_\vee$}
\UnaryInfC{$\vdash \bvar_3 \wedge \bvar_4 \Era
(\mathbf{1} \wedge \mathbf{2}) \vee
(\mathbf{3} \wedge \mathbf{4})$}
\RightLabel{$R^\Era_\cup$}
\BinaryInfC{$\vdash (\bvar_1\wedge \bvar_2) \vee
(\bvar_3 \wedge \bvar_4) \Era (\mathbf{1} \wedge \mathbf{2})
\vee (\mathbf{3} \wedge \mathbf{4})$}
\RightLabel{$R^\Era_\BOX$}
\UnaryInfC{$\vdash \top \Era \BOX^{1/3} (\mathbf{1} \wedge
\mathbf{2}) \vee (\mathbf{3} \wedge \mathbf{4})$}
\end{prooftree}
Indeed, $\mu\big(\model{(\bvar_1 \wedge \bvar_2) \vee
(\bvar_3 \wedge \bvar_4)}\big) = \frac{1}{2} \ge \frac{1}{3}$.
\end{proof}

\end{document}